\documentclass[aps,pra,showpacs,twocolumn,amsmath,amssymb]{revtex4}

\usepackage{epsfig,afterpage}
\usepackage{graphicx}
\usepackage{amsfonts}
\usepackage{amssymb}
\usepackage{indentfirst}
\usepackage{amsmath,amsthm}
\usepackage{dsfont}

\usepackage{epsfig}
\usepackage{subfigure}

\usepackage{multirow}
\usepackage{pstricks}
\usepackage{lscape}
\usepackage{verbatim,bbm}

\newtheorem{thm}{Theorem}
\newtheorem{lem}[thm]{Lemma}
\newtheorem{defn}[thm]{Definition}

\newcommand{\lra}{\longrightarrow}

\newcommand{\be}{\begin{equation}}
\newcommand{\ee}{\end{equation}}
\newcommand{\bea}{\begin{eqnarray}}
\newcommand{\eea}{\end{eqnarray}}
\newcommand{\1}{\mathbbm{1}}        
\newcommand{\tr}[1]{{\rm tr}\left[#1\right]}

\renewcommand{\>}{\rangle}
\newcommand{\<}{\langle}

\DeclareMathOperator{\trace}{tr}
\DeclareMathOperator{\rank}{rank}
\DeclareMathOperator{\range}{range}

\begin{document}

\title{A canonical form for Projected Entangled Pair States and applications}

\author{D. P\'erez-Garc\'ia$^1$, M. Sanz$^2$, C. E. Gonz\'alez-Guill\'en$^1$, M. M. Wolf$^3$, J. I. Cirac$^2$}
  \affiliation{$^1$Dpto. An\'alisis Matem\'atico and IMI, Universidad Complutense de Madrid, 28040 Madrid, Spain \\
  $^2$Max-Planck-Institut f\"ur Quantenoptik, Hans-Kopfermann-Str. 1, 85748 Garching, Germany \\
  $^3$Niels Bohr Institute, Blegdamsvej 17, 2100 Copenhagen, Denmark}

\begin{abstract}
We show that two different tensors defining the same translational
invariant {\it injective} Projected Entangled Pair State (PEPS) in a square lattice must be the same
up to a trivial gauge freedom. This allows us to characterize the
existence of any local or spatial symmetry in the state. As an application of these results we prove that a $SU(2)$ invariant PEPS with half-integer spin cannot be injective, which can be seen as a Lieb-Shultz-Mattis theorem in this context. We also give the natural generalization for $U(1)$ symmetry in the spirit of Oshikawa-Yamanaka-Affleck, and show that a PEPS with Wilson loops cannot be injective.
\end{abstract}


\maketitle

\begin{section}{Introduction}\label{sec:intro}
The isolation of Projected Entangled Pair States (PEPS) \cite{AKLT88,PEPS} as an
appropriate representation for ground states of 2D local
Hamiltonians \cite{Hastings} turns the problem of understanding 2D
quantum many body systems into the question: How can one characterize
the different phases of matter in terms of the tensors defining a
PEPS?

Though there are known examples of PEPS with topological order \cite{Verstraete, Guifre}, power law decay of correlations \cite{Verstraete}, $SU(2)$-symmetry \cite{AKLT88, Rico}, or universal power for measurement based quantum computation \cite{Verstraete, Gross}, characterizing these phases has turned out to be a daunting task. In this paper we provide a simple characterization of the existence of symmetries (both local and spatial) as a trivial consequence of the fact, which we call canonical form,  that two PEPS describing the same
translational invariant state in a square lattice are related by
invertible matrices in the virtual spins, as in Fig.
\ref{figtheorem}.

\begin{figure}[h!]
 \centering
  \includegraphics[width=0.5\textwidth]{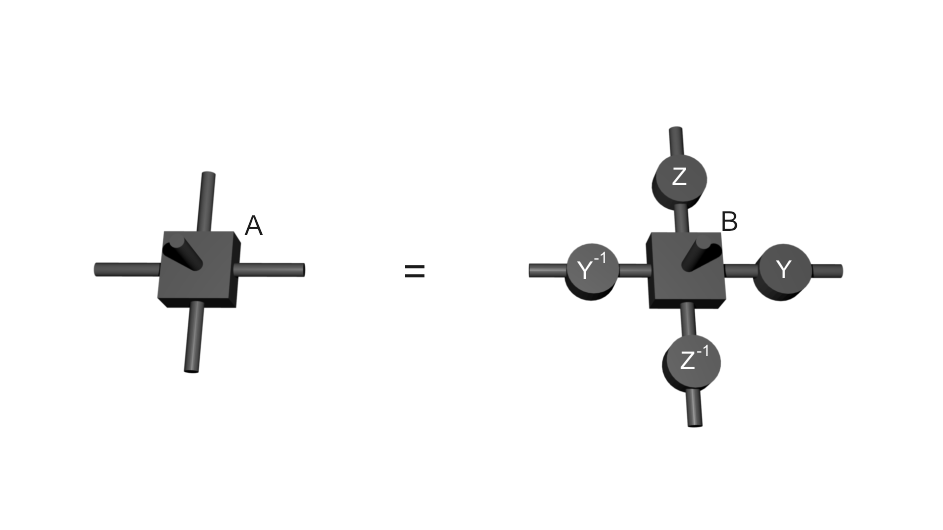}
  \caption{\footnotesize \emph{This is the Canonical Form for PEPS, i.e.,
  the relation that holds between the tensors
  which define the PEPS when they represent the same state.}}
  \label{figtheorem}
\end{figure}

This simple characterization illuminates the restrictions that symmetries impose on quantum systems. For instance one can in this context understand the validity of the Lieb-Shultz-Mattis theorem in arbitrary dimensions \cite{LSM61, HastingsLSM}, as well as its $U(1)$ generalization due to Oshikawa, Yamanaka and Affleck \cite{OYA97} (originally only in the 1D case). We can also understand why and how three of the main indicators of topological order, namely degeneracy of the ground state, existence of Wilson loops and correction to the area law, are related. Moreover, it has been proven in \cite{PWSVC08} that the existence of symmetries in increasing sizes of the system gives the appropriate definition of string order in 2D, overcoming the drawbacks sketched in \cite{Anfuso}. The importance of string orders in the study of quantum phase transitions, may vaticinate interesting applications in the future along this direction.

\

Before introducing PEPS formally, we will start with the simpler case of Matrix Product States (MPS), their 1D analogue \cite{FNW92, PVWC07}. Let us consider a system with periodic boundary conditions of $N$
(large but finite number of) sites, each of them with an associate
$d$-dimensional Hilbert space. An MPS on this
system is defined by a set of $D\times D$ matrices $\{ A_i \in \mathcal{M}_D
\, , \, i= 1, \ldots, d \}$  and reads
\begin{equation*}
|\phi_A \> = \sum_{i_1,\ldots, i_N} \tr{A_{i_1} \cdots A_{i_N}} |i_1 \cdots i_N \>\; .
\end{equation*}

An alternative but equivalent view is the valence bond construction: consider a pair of $D$ dimensional
ancillary/virtual Hilbert spaces associated to each site and connect every pair of neighboring virtual Hilbert spaces by maximally entangled states
(usually called entangled \textit{bonds}). The MPS is then the result of {\it projecting} the virtual Hilbert spaces into the real/physical one by the map
$\mathcal{A} =  \sum_{i \, \alpha \, \beta} A_{i , \alpha \beta}
|i\> \< \alpha \beta |$.

A key property within MPS theory is called {\it injectivity} \cite{FNW92, PVWC07} and it essentially means that different boundary conditions give rise to different states. Let us formally define it:
\begin{defn}[\it{Injectivity}]
An MPS $|\phi_A\>$ is {\it injective in a region $R$} (whose minimal length we denote by $L_0$) if the map $\Gamma_R(X)=\sum_{i_1,\ldots,i_{L_0}}\trace(XA_{i_1}\cdots
A_{i_{L_0}})|i_1\cdots i_{L_0}\>$ which associates boundary conditions of $R$ to states in $R$ is injective. That is, different boundary conditions give rise to different states. An MPS is said to be {\it injective} if it is injective for some region $R$.
\end{defn}

If we do not take into consideration translational invariance, we can talk about MPS with `open boundary conditions' (OBC). An OBC-MPS is then a state of the form
\begin{equation*}
|\Phi \> = \sum_{i_1,\ldots, i_N} A_{i_1}^{[1]} \cdots A_{i_N}^{[N]}
|i_1 \cdots i_N \>
\end{equation*}
where $A_i^{[k]}$ are $D_k \times D_{k+1}$ matrices with $D_1 =
D_{N+1} = 1$. By taking successive singular value decompositions one can always find  a {\it
 canonical} OBC-MPS form of a state \cite{Guifre2,PVWC07}, which is characterized by the following conditions:
\begin{enumerate}
\item $\sum_i A^{[m]}_iA^{[m]\dagger}_i=\mathbbm {1}$ for all
$1\le m\le N$. \item $\sum_i A^{[m]\dagger}_i\Lambda^{[m-1]}
A^{[m]}_i=\Lambda^{[m]}, $ for all $1\le m\le N$,  \item
$\Lambda^{[0]}=\Lambda^{[N]}=1$ and each $\Lambda^{[m]}$ is
diagonal, positive, full rank and $\trace{\Lambda^{[m]}}=1$.
\end{enumerate}

PEPS are the natural extension of the MPS beyond the 1D case, where the {\it projection} is
performed from a larger number of virtual Hilbert spaces depending
on the co-ordination number of the lattice (the square lattice, for instance,
 has four virtual Hilbert spaces). Therefore, the local
\emph{building blocks} are tensors instead of matrices which implies that most calculations
become much harder \cite{Norbert}.

Let us consider an $L \times N$ square lattice of spins of
dimension $d$. A PEPS consists on a tensor $A^{i}_{ldru}$ with 5 indexes: $i$ corresponding to the physical spin of
dimension $d$ and $l,d,r,u$ (left, down, right, up) corresponding to four virtual spaces of dimensions (\emph{bonds})
$D_1$ and $D_2$, as we did for MPS. The connections between two
sites are again performed by means of maximally entangled states
$|\Omega \> = \sum_{\alpha} |\alpha \alpha \>$. Then, the shape of
these states is
\begin{equation*}
|\phi_A \> = \sum_{i_1,\ldots,i_{NL}} \mathcal{C} ( A^{i}_{ldru}) |i_1 \ldots i_{NL} \>
\end{equation*}
where $\mathcal{C}$ means the contraction of all tensors
$A^{i}_{ldru}$ along the square lattice.

Associated to any PEPS $|\phi_A\>$ we can define a parent Hamiltonian $H_A$ \cite{PVWC08}, which is locally defined by the projector onto $\range(\Gamma_R)^\perp$. It is clear that the $|\phi_A\>$ is a ground state for $H_A$ and that it minimizes the energy locally, that is, $H_A$ is frustration free. In the case of 1D it is proven in \cite{FNW92,PVWC07} that a MPS is injective if and only if $|\phi_A\>$ is the unique ground state of $H_A$.

We can define the injectivity property for PEPS in the same way (see Fig \ref{fig:injectivity}). That is, the PEPS $|\phi_A\>$ is injective in a region $R$ if $\Gamma_R$ is injective. As in the 1D case it is clear that injectivity is a generic condition.
\begin{figure}[h!]
  \includegraphics[width=8cm]{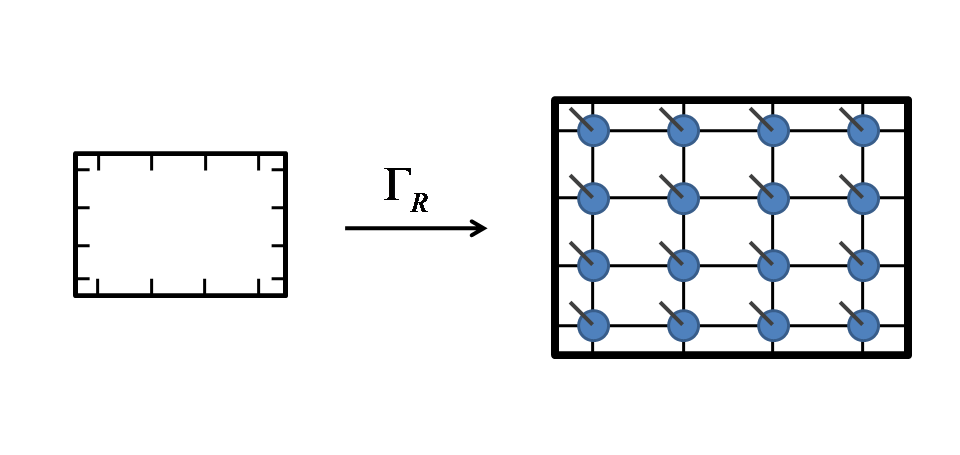}\\
  \caption{A PEPS is injective in a region $R$ if $\Gamma_R$ is injective, that is, if different boundary conditions give rise to different states in $R$.}\label{fig:injectivity}
\end{figure}

In the applications we will give below (Lieb-Shultz-Mattis, Wilson loops), the conclusion will often be that a given PEPS is not injective. What does this mean? As we list below, injectivity is closely related to uniqueness of the ground state of the parent Hamiltonian and to the saturation of the area law for the 0-Renyi entropy.
\begin{enumerate}
\item If a PEPS is injective, it is the unique ground state of its parent Hamiltonian \cite{PVWC08}.
\item If a PEPS is not-injective for any cylinder-shape region, any local frustration free Hamiltonian for which the given PEPS is a ground state has a degenerate ground space, as long as we grow one of the directions exponentially faster than the other. This is a trivial consequence of the 1D case \cite{PVWC07}.
\item The $0$-Renyi entropy of the reduced density matrix $\rho_R$ of a region $R$ of a PEPS with bond dimension $D$ is $\le |\partial R|\log D$. It is easy to see that if $S_0(\rho_R)=|\partial R|\log D$, then the PEPS is injective. That is, if a PEPS is not injective, there is a {\it correction} to the area law for the $0$-Renyi entropy.
\end{enumerate}

To finish this section we introduce the following notation. If $R$ is a region of the considered lattice underlying the PEPS, we denote by $A^R$  the joint tensor obtained after contracting all the tensors inside region $R$.
\end{section}


\begin{section}{The canonical form for MPS}

It is shown in \cite[Theorem 6]{PVWC07} that two injective
representations of the same MPS must be related by an invertible
matrix $R$ as $A_i=RB_iR^{-1}$. This holds if the number of sites satisfies
$N \ge 2 L_0 + D^4$, where $L_0$ is the size from which on one has
injectivity and $D$ is the bond dimension of the MPS. Since we are
interested (see the argument in Theorem \ref{thm:main} below) to
apply this to a ``column'' of a PEPS, the exponential dependence on
$D$ would be critical. So in this section, we modify the proof of
\cite[Theorem 6]{PVWC07} to make $N$ depend  on $L_0$ only. In
particular, we obtain that the result holds when $N\ge 4L_0+1$.

\begin{thm}\label{MPS}
Let
\begin{equation*}
|\psi_A\> = \sum_{i_1,\ldots,i_N = 1}^{d} {\rm  tr} (A_{i_1} \cdots A_{i_n})|i_1 \cdots i_N\>
\end{equation*}
and
\begin{equation*}
|\psi_B\> = \sum_{i_1,\ldots,i_N = 1}^{d} {\rm tr} (B_{i_1} \cdots B_{i_n})|i_1 \cdots i_N\>
\end{equation*}
be translational invariant MPS representations with bond dimension
$D$ which are injective for regions of size $L_0$. Then, if
$|\psi_A\> = |\psi_B\>$ and $N\geq 4L_0+1$, there exists an
invertible matrix $R$ such that $A_i = R B_i R^{-1}$, for all $i$.
\end{thm}

\begin{proof}
We can obtain an OBC representation by noticing that

\begin{equation*}
|\psi_A \> = \sum_{i_1,\ldots ,i_N=1}^{d}a_{i_1}^{[1]}(A_{i_2}\otimes \1 )
\cdots (A_{i_{N-1}} \otimes \1 )a_{i_N}^{[N]}|i_1 \cdots i_N\>
\end{equation*}
where $a_i^{[1]}$ is the vector that contains all the rows of $A_i$
and $a_i^{[N]}$ is the vector that contains all the columns in
$A_i$. Doing the same with the B's
\begin{equation*}
|\psi_B\>=\sum_{i_1,\ldots,i_N = 1}^{d}b_{i_1}^{[1]}(B_{i_2} \otimes \1 )
\cdots (B_{i_{N-1}}\otimes \1 )b_{i_N}^{[N]}|i_1 \cdots i_N\>
\end{equation*}

Getting from them an OBC canonical representation (with matrices
$C$'s for the $A$'s and matrices $D$'s for $B$'s) as in
\cite[Theorem 2]{PVWC07} we obtain $Y_j^1$, $Z_j^1$, $Y_j^2$ and $Z_j^2$  with
$Y_j^1Z_j^1=\1$, $Y_j^2 Z_j^2=\1$ such that

\begin{gather*}
C_i^{[1]}=a_i^{[1]} Z_1^1, C_i^{[N]}=Y^1_{N-1}a_i^{[N]}\\
C_i^{[m]}=Y^1_{m-1} (A_i\otimes \1)  Z^1_m \text{ for } 1 < m < N\\
D_i^{[1]}=b_i^{[1]} Z_1^2, D_i^{[N]}=Y^2_{N-1}b_i^{[N]}\\
D_i^{[m]}=Y^2_{m-1} (B_i\otimes \1) Z^2_m \text{ for } 1 < m < N
\end{gather*}

Besides using theorem 3.1.1' in \cite{Horn}, we get that any two OBC
canonical representations are related by unitaries, that is, there
exists $V_1,...,V_{N-1}$ such that

\begin{gather*}
C_i^{[1]}V_1= D_i^{[1]}, V_{N-1}C_i^{[N]}= D_i^{[N]}\\
V_{j-1}C_i^{[j]}V_j= D_i^{[j]} \text{ for } 1<j<N
\end{gather*}

Now, by using injectivity as in \cite[Theorem 6]{PVWC07}, we know
that $Z^r_s, Y^r_s$ are invertible for $r=1,2$ and $L_0\le s\le
N-L_0$ and so are the $D^2\times D^2$ matrices $W_k$ defined as

\begin{equation*}
W_{k}=Z^2_{L_0+k}V_{L_0+k}Y_{L_0+k}^1\;\; k=0,...,2L_0+1.
\end{equation*}

It is easy to verify that for all $i$,

\begin{equation*}
W_k(A_i\otimes \1)W_{k+1}^{-1}=(B_i\otimes\1) \text{ for } 0\leq k\leq 2L_0.
\end{equation*}

In fact, by grouping and denoting $A_{I_l}=A_{i_1} \cdots A_{i_l}$,
we have that

\begin{equation}\label{eq1}
W_m(A_{I_{n-m}}\otimes \1) W_{n}^{-1} = B_{I_{n-m}}\otimes \1
\end{equation}
for every $0 \leq m < n \leq 2L_0 + 1$ and every multi-index
$I_{n - m}$. Then for suitable values of $m$ and $n$, we obtain
\begin{equation*}
W_{k+1}^{-1}W_k(A_{I_{2L_0-k}}\otimes \1)W_{2L_0}^{-1}W_{2L_0+1}=A_{I_{2L_0-k}}\otimes \1
\end{equation*}
for every $0\leq k\leq L_0$.

As we are in an injective region for every $k$, the matrix could be
taken as the identity and then we get that
\begin{equation}\label{eq2}
T := W_{k+1}^{-1}W_k = W_{2L_0+1}^{-1}W_{2L_0}
\end{equation}
for every $0\leq k\leq L_0$.

Therefore, $T(X\otimes \1)T^{-1}=(X\otimes \1)$ for every $X$. Let
us make use of the following lemma, which is a consequence of
\cite[Theorem 4.4.14]{Horn}:

\begin{lem} \label{lemma0}
If $B$, $C$ are squares matrices of the same size $n\times n$, the
space of solutions of the matrix equation
\begin{equation*}
W(C\otimes \1)=(B\otimes\1)W
\end{equation*}
is $S\otimes M_n$ where $S$ is the space of solutions of the
equation $XC=BX$.
\end{lem}

With this at hand it is easy to deduce that $T=\1\otimes Z$ so that
\begin{equation*}
W^{-1}_{L_0}W_{0} = W^{-1}_{L_0}W_{L_0-1}W^{-1}_{L_0-1} \cdots W_{0} = (\1\otimes Z)^{L_0}
\end{equation*}
from where we obtain
\begin{equation*}
W^{-1}_{L_0} = (\1\otimes Z^{L_0})W_{0}^{-1}
\end{equation*}
and in the same way
\begin{equation*}
W^{-1}_{L_0+1} = (\1\otimes Z^{L_0+1})W_{0}^{-1}\; .
\end{equation*}

Replacing in Eq. (\ref{eq1})

\begin{equation*}
\begin{aligned}
(B_{I_{L_0}}\otimes \1) &=W_{0}(A_{I_{L_0}}\otimes \1)W^{-1}_{L_0}\\
&=W_{0}(A_{I_{L_0}}\otimes Z^{L_0})W_{0}^{-1}\\
(B_{I_{L_0+1}}\otimes \1) &=W_{0}(A_{I_{L_0+1}}\otimes \1)W^{-1}_{L_0+1}\\
&=W_{0}(A_{I_{L_0+1}}\otimes Z^{L_0+1})W_{0}^{-1}
\end{aligned}
\end{equation*}

By using injectivity of $B_{I_{L_0}}$ and $B_{I_{L_0+1}}$, we
can sum with appropriate coefficients to obtain $\1$ on the LHS.
Then, we get that  $Z^{L_0}=\1=Z^{L_0+1}$, which gives $Z=\1$ and
hence $B_i\otimes \1 = W_0(A_i \otimes \1)W_0^{-1}$ for all $i$.

By \cite[Theorem 4 and Proposition 1]{PVWC07}, we can assume w.l.o.g.
that $\sum_{i}A_iA_i^\dagger=\1$ and that $\sum_iB_i^\dagger \Lambda
B_i=\Lambda$ for a full-rank diagonal matrix $\Lambda$. The proof
follows straightforwardly from here as in \cite[Theorem 6]{PVWC07}.
\end{proof}
\end{section}


\begin{section}{The canonical form for PEPS}

In this section, we show that Theorem \ref{MPS} holds in any spatial
dimension: two injective representations of the same PEPS are
related by the trivial gauge freedom in the bonds (Fig.
\ref{figtheorem}).

We prove the result in 2D by using the result in 1D, and the argument
can be generalized to larger spatial dimensions by induction. We will initially consider a square
lattice, but we show at the end of the section how to
extend the result to the honeycomb lattice.

\begin{thm}\label{thm:main}
Let $|\psi_A\>$ and $|\psi_B\>$ be two PEPS
in a $L\times N$ square lattice given by tensors $A=\sum_{i;ldru}
A^i_{ldru}|i\>\<ldru|$, $B=\sum_{i;ldru} B^i_{ldru}|i\>\<ldru|$ with
the property that for a region of size smaller than $L/5\times N/5$
both PEPS are injective. Then $|\psi_A\>=|\psi_B\>$ if and only if
there exist invertible matrices $Y,Z$ such that $A^i(Y\otimes Z
\otimes Y^{-1}\otimes Z^{-1})=B^i$ for all $i$ (Fig.
\ref{figtheorem}). Moreover $Y$ and $Z$ are  unique.
\end{thm}

The uniqueness is a simple consequence of injectivity. For the existence part, let us split the proof into a sequence of lemmas, in order to make
it clearer.

\begin{lem}\label{lema2}
If a region of size $H\times K$ of a translational invariant PEPS is
injective, the same happens for a region of size $(H+1)\times K$
(and $H\times (K+1)$)
\end{lem}

\begin{proof}
Note that the region of size $1 \times K$ is injective when the
upper and the physical system are considered as inputs (left picture
of Fig. \ref{fig8}). To see this, take an injective region $S$ of
dimension $H \times K$ and split it into two subregions, as in the
right picture of Fig. \ref{fig8} with $T=H-1$. For simplicity in the
rest of the proof we gather the indexes $u^1$, $u^2$, $u^3$ and
$d^1$, $d^2$, $d^3$ and call them $u$ and $d$ respectively.

\begin{figure}[h!]
 \centering
 \includegraphics[width=7cm]{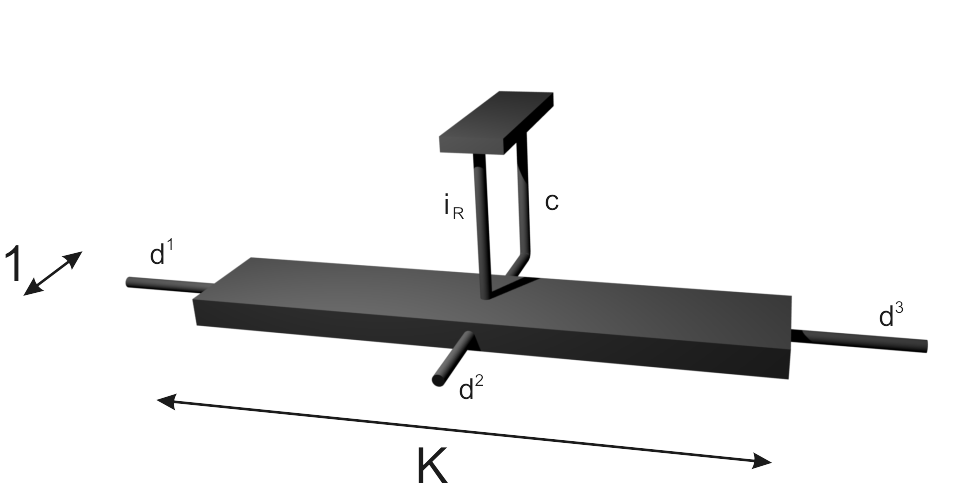}\\
 \includegraphics[width=7cm]{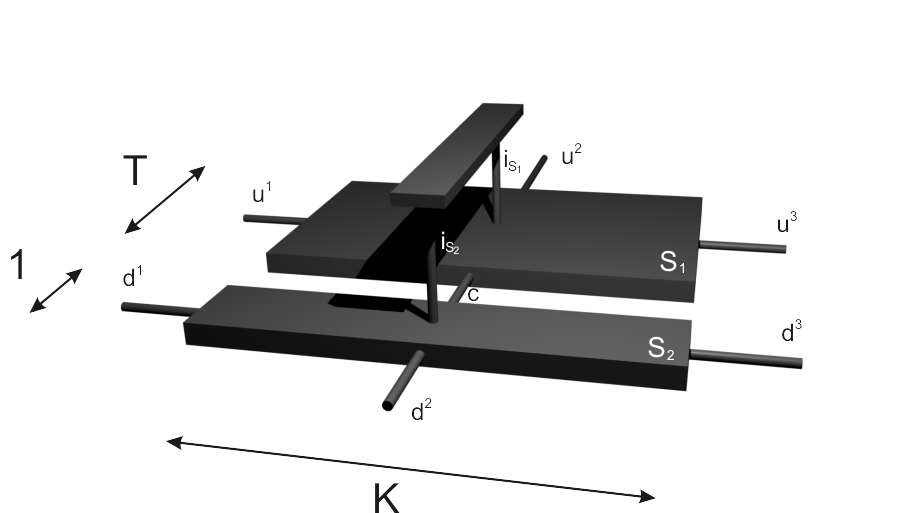}
 \caption{\footnotesize \emph{This figure represents the argument used to prove Lemma \ref{lema2}.}}
 \label{fig8}
\end{figure}

Using injectivity of the region $S$, there exists
$\{\alpha_{i_{S_1},j_{S_2}}^{u_0, d_0}\}_{i_{S_1},j_{S_2}}$ for any
$u_0,d_0$ such that
\begin{equation*} \sum_{c, i_{S_1},j_{S_2}}
\alpha_{i_{S_1},j_{S_2}}^{u_0, d_0} A_{i_{S_1}}^{u, c}
A_{j_{S_2}}^{c, d}=\delta_{u,u_0} \delta_{d,d_0}
\end{equation*}
Taking $u=u_0$ we get
\begin{equation*}
\sum_{c,i_{S_1},j_{S_2}} \alpha_{i_{S_1},j_{S_2}}^{u_0,
d_0} A_{i_{S_1}}^{u_0, c} A_{j_{S_2}}^{c, d}= \delta_{d,d_0}
\end{equation*}

Now, if we take a region $S$ of size $(H+1) \times K$ and divide it
as in Fig. \ref{fig8} with $T=H$, there exists
$\{\beta_{c, j_{S_2}}^{ d_0}\}_{c, j_{S_2}}$ for any $d_0$ such that

\begin{equation*}
\sum_{c,j_{S_2}} \beta_{c, j_{S_2}}^{d_0} A_{j_{S_2}}^{c, d}=
\delta_{d,d_0}
\end{equation*}
By using injectivity of a region
of dimension $H \times K$, there exists
$\{\alpha_{i_{S_1},j_{S_2}}^{u_0,c_0,d_0}\}_{i_{S_1}}$ such that
\begin{equation*}
\sum_{i_{S_1}} \alpha_{i_{S_1},j_{S_2}}^{u_0, c_0, d_0}
A_{i_{S_1}}^{u, c} = \beta_{c_0, j_{S_2}}^{ d_0}\delta_{u,u_0}
\delta _{c,c_0}
\end{equation*}
By putting both equalities together, we find
\begin{gather*}
\sum_{c,c_0,i_{S_1},j_{S_2}} \alpha_{i_{S_1},j_{S_2}}^{u_0,c_0, d_0}
A_{i_{S_1}}^{u, c} A_{j_{S_2}}^{c, d} =\\ =\sum_{c,c_0 j_{S_2}}
\beta_{c_0, j_{S_2}}^{ d_0}\delta_{u,u_0} \delta _{c,c_0}
A_{j_{S_2}}^{c, d}=\\
=\sum_{c_0 j_{S_2}} \beta_{c_0, j_{S_2}}^{ d_0}\delta_{u,u_0}
A_{j_{S_2}}^{c_0, d}=\delta_{u,u_0} \delta _{d,d_0}
\end{gather*}
and so $S$ is an injective region.
\end{proof}

This allows us to reduce the 2D case to the 1D case by grouping all
the tensors in a column. The 1D case (Theorem \ref{MPS}) ensures
that there is a global invertible matrix $Y$ which verifies the
equality in Fig. \ref{fig1}. Now

\begin{figure}[h!]
 \centering
  \includegraphics[width=0.5\textwidth]{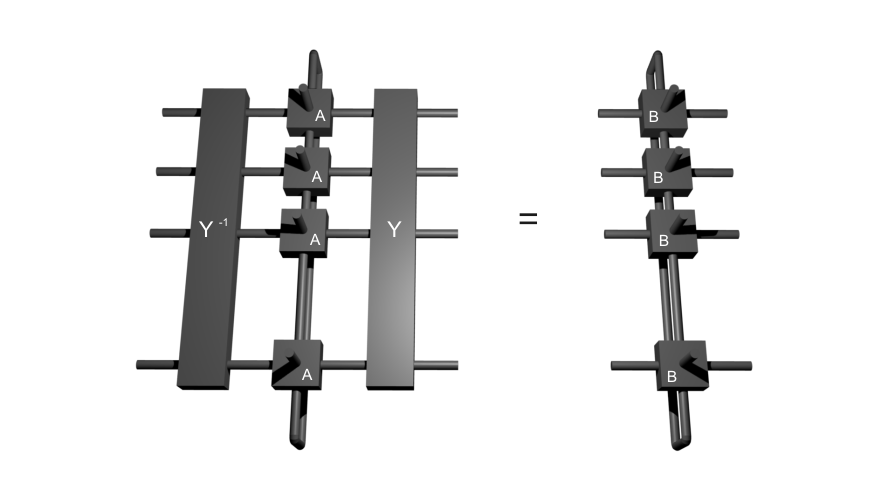}
  \caption{\footnotesize \emph{Translational invariance
  and  injectivity allow to reduce the 2D case to the 1D case. }}
 \label{fig1}
\end{figure}

\begin{lem}\label{lema3}
$Y$ maps product vectors into product vectors.
\end{lem}

We will show that $Y$ maps any product vector to a vector with the
property:

(*) {\it It is a product in {\it any} bipartition R-S, for regions R
and S of consecutive spins and size $\ge L/5$.}

Since any vector with property (*) is trivially a product vector,
this would finish the proof. So let us take a product
$\otimes_i|x_i\>$  and assume that this product is mapped by $Y$
into a vector that can be written in some orthonormal bases as $Y
(\otimes_i |x_i \>) = \sum_{r=1,2,\ldots} \beta_r |v_rw_r\>$ in a
partition R-S for regions of consecutive spins and size $\ge L/5$. For the same bipartition, we
may write $\otimes_i\<x_i|Y^{-1}=\sum_{r=1,2,\ldots} \alpha_{r}
\<v'_{r}w'_r|$, which could be a product. We group $N/5$ columns,
sandwich with $\otimes_i|x_i\>$ in Fig \ref{fig1} and analyze the
Schmidt rank between the two physical $R\times N/5$ and $S\times
N/5$  systems in both the right and left part of Fig. \ref{fig1}. It
clearly gives $D^{2N/5}$ in the RHS by using injectivity.  By
performing the changes of bases $|r\>\mapsto |v_r\>$ and
$|r\>\mapsto |w_r\>$ (and the same for the primes) to the tensors
$A^{R\times N/5}$ and $A^{S\times N/5}$ in the LHS, it gives new
tensors $A'$ and $A''$ for which we get

\begin{equation*}
\sum_{abcd}\alpha_a{\beta}_c
[\sum_{i}A'^i_{abcd}|i\>][\sum_jA''^j_{adcb}|j\>]
\end{equation*}

By means of injectivity, we know that the set
$\{\sum_{i}A'^i_{abcd}|i\>\}_{abcd}$ is linearly independent (and
the same for $A''$). This means that the Schmidt rank of the LHS is
at least $2D^{2N/5}$, which is the desired contradiction.

The following three lemmas specify the form of $Y$:
\begin{lem} \label{lema4}
If $Y$ is invertible and takes products to products it is of the
form $P_{\pi} (Y_1\otimes\cdots\otimes Y_L)$ where $P_\pi$
implements a permutation $\pi$ of the Hilbert spaces.
\end{lem}

\begin{proof}
We reason for simplicity in the bipartite case---the argument
generalizes straightforwardly to the general case by induction. Let $Y:
\mathbbm{C}^d \otimes \mathbbm{C}^d \lra \mathbbm{C}^d \otimes
\mathbbm{C}^d$ be invertible which takes products to products and
denote $\{|i,j\>\}_{i,j=1,...,d}$ the product basis. Let
$Y(|i,1\>)=|\alpha_i,\beta_i\>$. Take $i_0 \neq i_1 \in \{1,...,d\}$,
then
$Y(|i_0,1\>+|i_1,1\>)=|\alpha_{i_0},\beta_{i_0}\>+|\alpha_{i_1},\beta_{i_1}\>$
is a product and, as $Y$ is invertible, then either I)
$\alpha_{i_0}\propto\alpha_{i_1}$ \& $\beta_{i_0}\not \propto \beta_{i_1}$ or II)
$\alpha_{i_0}\not\propto\alpha_{i_1}$ \& $\beta_{i_0}\propto \beta_{i_1}$, where $\propto$ means proportional to. In
fact, we are always in the same case: if $d=2$ there is only one
case, otherwise take three distinct $i_0,i_1, i_2 \in \{1,...,d\}$
such that $\alpha_{i_0}\not\propto\alpha_{i_1}$ and $\beta_{i_1}\not\propto
\beta_{i_2}$ then we get a contradiction from the fact that
$Y(|i_0,1\>+|i_2,1\>)$ is a product.

The same argumentation can be carried out for the second tensor. We
can therefore assume w.l.o.g. that
\begin{equation*}
Y(|i,1\>)=|\alpha_i,\beta_1\>
\end{equation*}
and
\begin{equation*}
Y(|1,j\>)=|\alpha_1,\beta_j\>
\end{equation*}
In the other case, we just permute the indexes by means of the swap operator $P_{\pi}$.

Let us consider $Y(|i,j\>)=|a_{i,j},b_{i,j}\>$. Now, since
\begin{equation*}
Y(|i,j\>+|i,1\>)=|\alpha_i,\beta_1\>+|a_{i,j},b_{i,j}\>
\end{equation*}
is a product, we obtain that $\alpha_i\propto a_{i,j}$ or $\beta_1\propto b_{i,j}$.
However, the second case is only possible if $j=1$ because of the
invertibility of $Y$, and then $a_{i,j}\propto\alpha_i$. A similar
argumentation over the second tensor gives
$Y(|i,j\>)=c_{i,j}|\alpha_i,\beta_j\>$. Now making $Y(\sum_{i,j=1}^d|ij\>)=\sum_{i,j=1}^d c_{i,j}|\alpha_i,\beta_j\>$ and knowing that the Schmidt rank of the resulting vector must be one, we conclude that the matrix $(c_{i,j})_{i,j}$ has rank one and therefore is of the form $c_{i,j}=r_is_j$ giving $Y(|i,j\>)=|r_i\alpha_i,s_j\beta_j\>$, the desired result.
\end{proof}

Let us now show that $P_\pi$ is the
trivial permutation:
\begin{lem}
$P_\pi=\1$
\end{lem}
\begin{proof}
Assume that $P_\pi$ is not the identity. Take a $R-S$ bipartition (with sizes $\ge L/5$) such that
$P_\pi$ maps one Hilbert space of $R$ into one of $S$. We block again
$N/5$ columns to get two injective $R\times N/5$ and $S \times N/5$
regions. Denoting by $R_1$ and $S_1$  the parts of the regions that stay within the regions and by $R_2$, $S_2$  the ones that are mapped to the other side, we can
decompose $Y$ as in Fig. \ref{fig2}.

\begin{figure}[h!]
 \centering
  \includegraphics[width=0.5\textwidth]{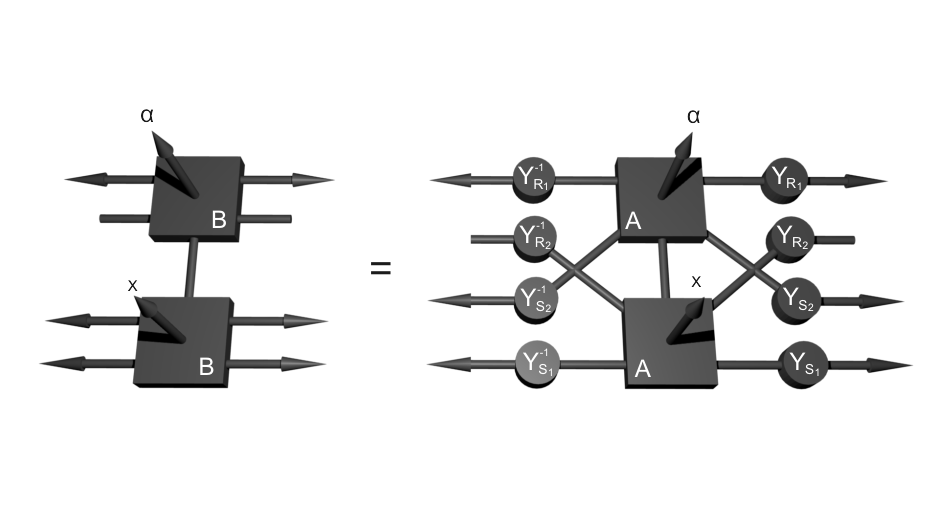}
  \caption{\footnotesize \emph{The cones represent vectors
  multiplying the legs of the tensor. In the virtual space, these vectors
  are $|0\>$, while the vectors in the leg corresponding to the real space are $|x\>$ and $|\alpha\>$ respectively.}}
 \label{fig2}
\end{figure}

Consider now Fig.\ref{fig2}. We contract all virtual indices but the pair in the second row with $|0\rangle$ and the physical indices with $|\alpha\rangle$ and $|x\rangle$ where the latter is chosen such that $A|x\>=|0\>|0\>|0\>|0\>$. Let $V$ be the span in the remaining two virtual indices under the variation of $|\alpha\rangle$.
 It is clear that
in the LHS of Fig.\ref{fig2}, $\dim V=\dim {({\rm support}{(Y_{R_2})})}$, whereas in
the RHS $\dim V=1$, which leads to a contradiction unless $R_2$
and $S_2$ are empty.
\end{proof}

By using both, injectivity and  translational invariance of
the RHS in Fig. \ref{fig1}, we observe that

\begin{lem}
$Y_i =  Y$ for all $i$.
\end{lem}

We redefine now $A^i$ as $\sum_{ldru} A^i_{ldru}(Y^{-1}\otimes
\1)|ld\>\<ru|(Y\otimes \1)$, block $N/5$ columns together and
sandwich with $|nn\cdots n\>$ and $\<mm\cdots m|$ in Fig.
\ref{fig1}. Defining $\tilde{A}^{i;mn}$ as
\begin{equation*}
\sum_{bd}\tilde{A}^{i; mn}_{bd}|b\>\<d|=
\sum_{bd}\<m|A^{1\times N/5}|n\>|b\>\<d|
\end{equation*}
and the analog for $\tilde{B}^{i;mn}$, we have two injective
representations of the same MPS. By means of the 1D case (Theorem
\ref{MPS}), we obtain invertible matrices $Z_{nm}$ such that
$Z^{-1}_{mn}\tilde{A}^{i;mn}Z_{mn}=\tilde{B}^{i;mn}$.

The next step is to show that $Z_{mn}$ does not depend on $m$ and
$n$. We sandwich in Fig. \ref{fig1} with  $\<{m'}|^{\otimes L/2}
\<m|^{\otimes L/2}$ and $|{n'}\>^{\otimes L/2}|n\>^{\otimes L/2}$
and get Fig. \ref{fig3}. By summing with appropriate coefficients in
order to obtain "deltas", we obtain that
$\<l|Z_{mn}Z^{-1}_{m'n'}|k\>\<r|Z_{mn}^{-1}Z_{m'n'}|s\>=\delta_{kl}
\delta_{rs}$,
so $Z^{mn}=Z$ is indeed independent of $m$ and $n$. By reasoning as
above in the other direction, one can prove that $Z=Z'^{\otimes
N/5}$.

\begin{figure}[h!]
 \centering
  \includegraphics[width=0.5\textwidth]{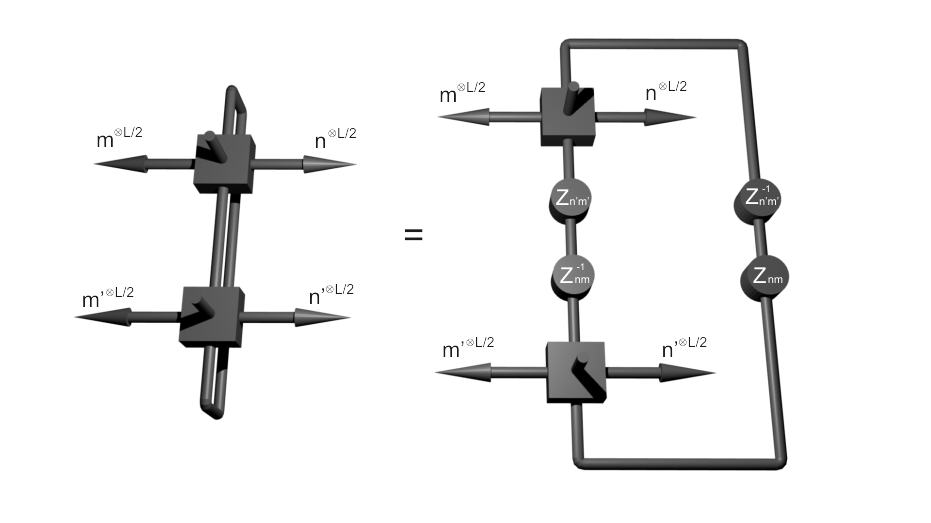}
  \caption[Fig. 5]{\footnotesize \emph{This is a representation of
  the resulting tensor after blocking $\frac{L}{2} \times \frac{N}{5}$
  sites. The vectors in the virtual space  correspond to tensor
  products of $\frac{L}{2}$ local vectors.}}
 \label{fig3}
\end{figure}

Up to now, we have proven the following Lemma

\begin{lem}
For any length $H$ for which one gets injectivity in the orthogonal
direction, we have the structure shown in Fig. \ref{fig4}. The case
where vertical is interchanged by horizontal is equivalent.
\end{lem}

\begin{figure}[h!]
 \centering
  \includegraphics[width=0.5\textwidth]{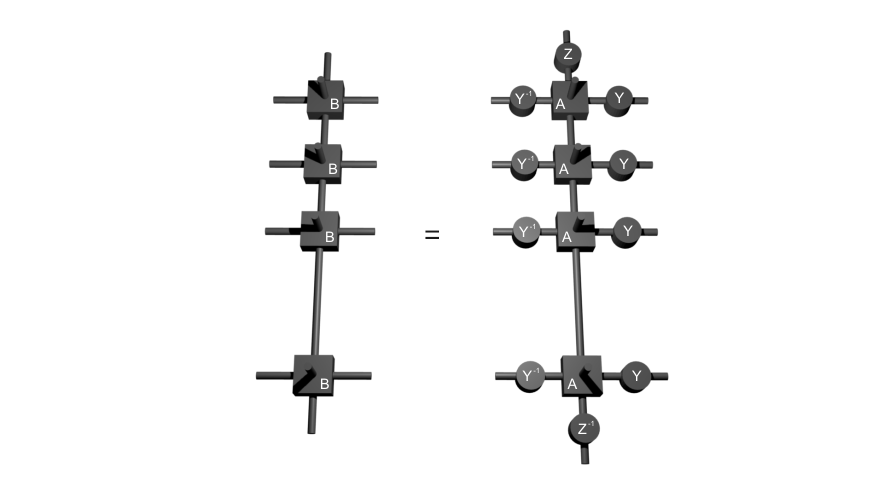}
  \caption[Fig. 6]{\footnotesize \emph{A subsystem of $H$ spins
  has two equivalent injective
  TI-PEPS representations iff they are connected by invertible matrices. }}
 \label{fig4}
\end{figure}

We want to prove now Theorem \ref{thm:main}. Let us
consider a $H\times K$ injective region, for instance $H = L/5$,
$K = N/5$. From Lemma \ref{lema2}, the larger regions in Fig.
\ref{fig5} are also injective. If we replace Fig. \ref{fig4}
first in each subregion (not the center), and then in the whole
region, we get the
desired result by using injectivity in the four subregions.

\begin{figure}[h!]
 \centering
  \includegraphics[width=0.3\textwidth]{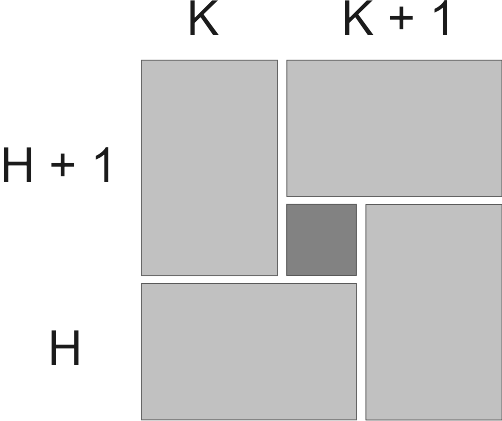}
  \caption[Fig. 7]{\footnotesize \emph{Representation of the regions
  of injectivity for the proof of Theorem \ref{thm:main}.}}
 \label{fig5}
\end{figure}

\

As we said in the introduction of this section, we can generalize
Theorem \ref{thm:main} to the honeycomb lattice. We need to
prove first the following

\begin{lem} \label{split}
Let $A,C \in \mathcal{M}_{d_1,d_2}$ and $B,D \in
\mathcal{M}_{d_2,d_3}$ and let us assume that $\min (d_1,d_2,d_3) =
d_2$. Then, if $A B = C D$ and $\rank (B) = \rank(D) = d_2$ there
exists an invertible matrix $W$ such that $A = C W$ and $B = W^{-1}
D$.
\end{lem}
\begin{proof}
Since $B$ is full-rank and $\min (d_1,d_2,d_3) = d_2$, there exists a
matrix that we can call $B^{-1}$ such that $B B^{-1} = \1_{d_2}$.
Therefore, $A = C (D B^{-1})$ and we can denote $W = D B^{-1}$,
which is an invertible matrix. Similarly $B=A^{-1}CD$ and we can denote $U=A^{-1}C$. Since $UW=A^{-1}CDB^{-1}=BB^{-1}=\1_{d_2}$, we get that $U=W^{-1}$ and hence
$B = W^{-1} D$.
\end{proof}
We can  now prove the theorem for the honeycomb lattice. Let us
remark that the unit cell of this lattice contains two
sites and that the lattice associated to the unit cells is a
square lattice. The translational invariance is not site
by site, but unit cell by unit cell.

\begin{thm}[The honeycomb lattice]
Let $|\Psi \>$ and $|\Psi'\>$ be two PEPS defined in a
honeycomb lattice and such that the square lattice constituted by
the unit cells fulfils the conditions of Theorem
\ref{thm:main}. Then, $|\Psi \> = |\Psi'\>$ iff the
conditions shown in Fig. \ref{fig9} hold.
\end{thm}

\begin{figure}[h!]
 \centering
  \includegraphics[width=0.5\textwidth]{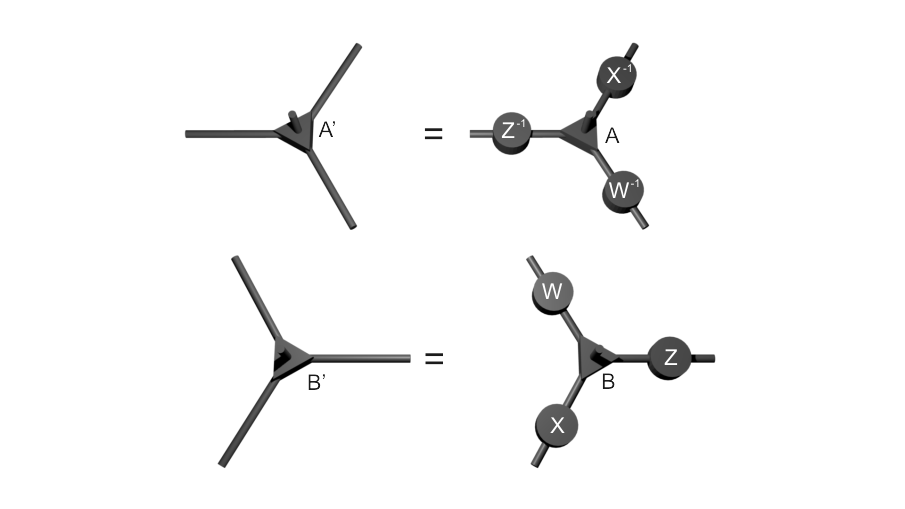}
  \caption{\footnotesize \emph{These are the relations which the
  tensors defining two TI-PEPS on a honeycomb lattice must fulfil
  in order to represent the same state.}}
 \label{fig9}
\end{figure}

\begin{proof}
Let us apply Theorem \ref{thm:main} to the square lattice which the
unit cell constitutes. Then, we obtain the equality shown in Fig.
\ref{fig10} and Lemma \ref{split} completes the proof of the theorem.
\end{proof}

\begin{figure}[h!]
 \centering
  \includegraphics[width=0.5\textwidth]{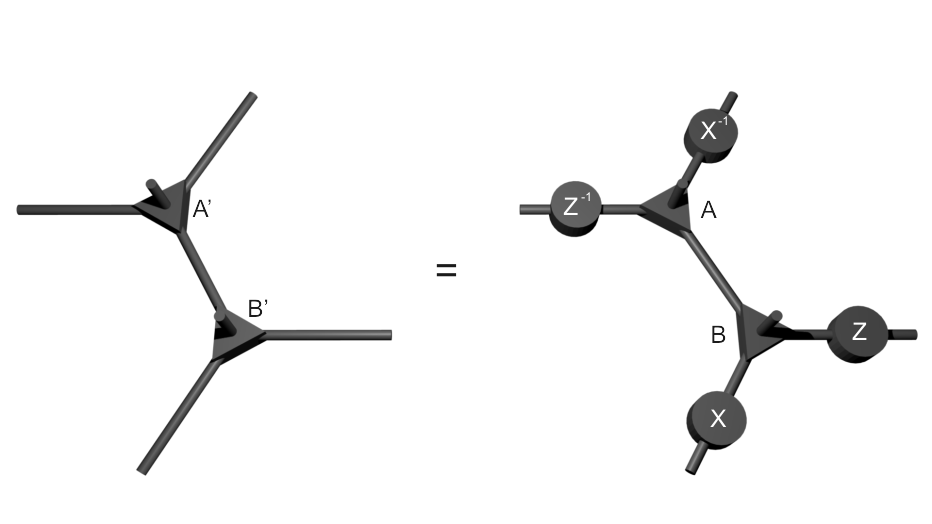}
  \caption{\footnotesize \emph{The possibility of transforming
  the honeycomb lattice into a square lattice by blocking tensors enables us to apply the result on equivalent TI-PEPS representations for the square lattice.}}
 \label{fig10}
\end{figure}

\end{section}


 \begin{section}{Symmetries}

String order parameters have been proven to be a very useful tool in the detection and understanding of quantum phase transitions. However, as pointed out in \cite{Anfuso} its application could not go beyond the 1D case. In \cite{PWSVC08}, with the aid of MPS, it has been shown that the existence of a string order parameter is intimately related to the existence of a symmetry, which allows to design an appropriate 2D definition: the existence of a local symmetry when we consider increasing sizes of the system. A trivial sufficient condition for this to hold in PEPS is proposed there (see Fig. \ref{fig6}), and further analyzed in \cite{Vidal-symm} in the more general context of Tensor Network States. The aim of this section is to prove that, for injective PEPS, the condition is also necessary. The 1D version is proven in \cite{PWSVC08} with the assumption of injectivity and in \cite{SWPC09} for the  general 1D case.

\begin{thm}[Local symmetry]\label{symm}
If a TI-PEPS defined on an $L\times N$ lattice has a symmetry
$u$, i.e. $u^{\otimes NL}|\psi_A\>=e^{i\theta'}|\psi_A\>$, and is injective in regions of size $L/5\times N/5$, then the tensors
defining it satisfy the relation in Fig. \ref{fig6} with $e^{i\theta
NL}=e^{i\theta'}$. Moreover, if $u_g$ is a representation of a
continuous group $G$, then $Y_g$, $Z_g$ and $e^{i\theta_g}$ are  representations as well.
\end{thm}

\begin{figure}[h!]
 \centering
  \includegraphics[width=0.5\textwidth]{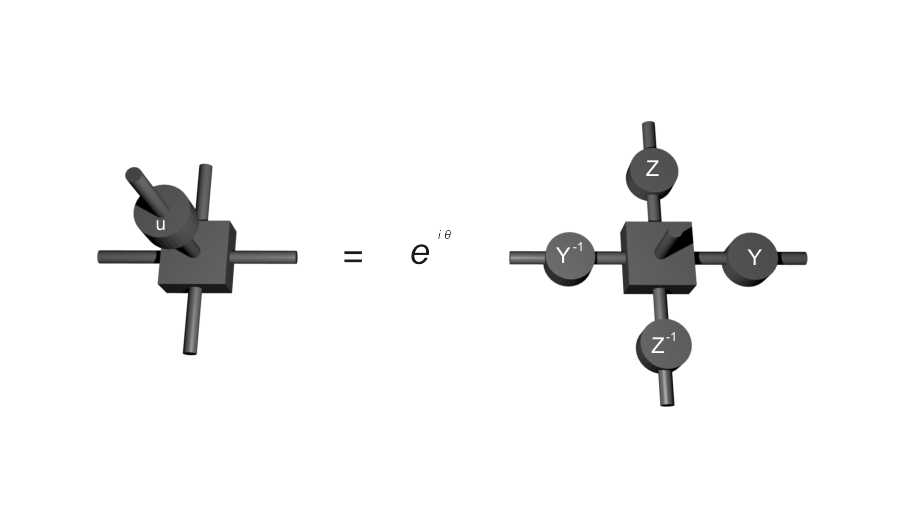}
  \caption{\footnotesize \emph{This is a graphical
  representation of the equation that a PEPS fulfils if it is
  invariant under a representation $u_g$ of a group $G$. Then,
  the symmetry is inherited into a couple of representations of
  $G$, called $Y_g$ and $Z_g$, up to a phase $e^{i \theta_g}$.}}
 \label{fig6}
\end{figure}

\begin{proof}
Notice that when acting with $u$ and $e^{-i\theta}$ on the tensor $A$ which defines the PEPS (see Fig \ref{fig6}), we get a new tensor $B$ that is also injective in regions of size $L/5\times N/5$ and such that $|\psi_A\>=|\psi_B\>$. Theorem \ref{thm:main} then gives  the result. In order to prove that
the invertible matrices $Y_g$ and $Z_g$ are representations of $G$, we only
need to follow the arguments used in \cite[Theorem 7]{SWPC09}.
\end{proof}

With exactly the same reasoning, we can characterize the spatial symmetries: reflection,
$\pi / 2$-rotations and $\pi$-rotations:

\begin{thm}[Reflection symmetry]\label{symm: refl}
Let us consider an $L \times N$ TI-PEPS with the property that for a
region of size smaller than $L/5 \times N/5$ it is injective. If this
PEPS is invariant under a reflection with respect to a vertical
axis, then there exist invertible matrices $Y$, $Z$ such that the tensors defining the PEPS verify Fig.
\ref{fig11}, that is, $\sum_{ldru} A^i_{ldru} \<ldru|= (\sum_{ldru}
A^i_{ldru}\<rdlu|) Y\otimes Z\otimes Y^{-1}\otimes Z^{-1}$
for all $i$.
\end{thm}

\begin{figure}[h!]
 \centering
  \includegraphics[width=0.5\textwidth]{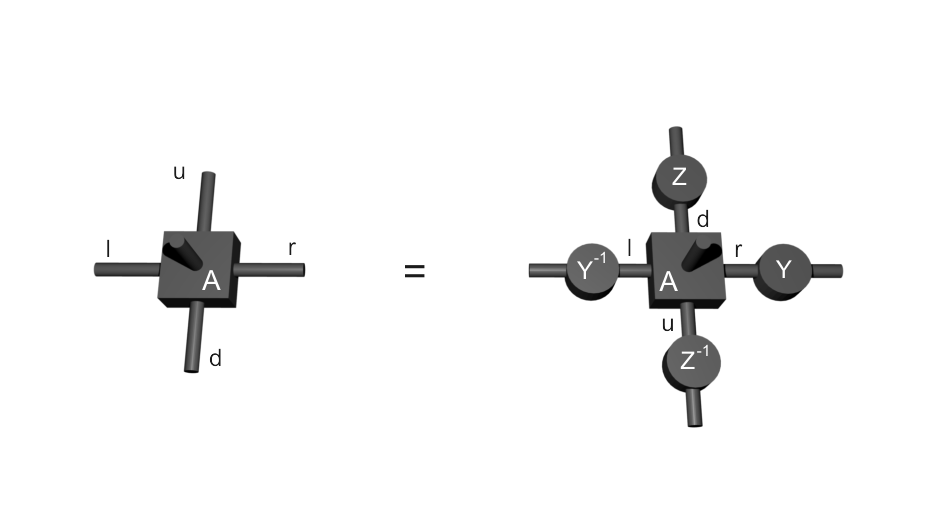}
  \caption{\footnotesize \emph{This figure represents the
  condition which must be fulfilled by a TI-PEPS in order to generate
  a state invariant under reflections (in this case with respect to
  the horizontal plane).}}
 \label{fig11}
\end{figure}

Moreover, it is easy to see that $Y,Z$ must satisfy $Y^T=Y$, $Z^2=\1$.
The characterization of the reflection with respect to the
horizontal axis follows straightforwardly by changing the roles of the
horizontal/vertical directions.

\begin{thm}[spatial $\pi/2$-rotation symmetry]\label{symm: pi2rot}
If an $L\times N$ TI-PEPS with the property that for a region of
size smaller than $L/5\times N/5$ it is injective has a spatial $\pi /
2$-rotation invariance, then there exist invertible matrices $Y$,
$Z$ such that the tensors $A^i$
defining the PEPS verify Fig. \ref{fig12}, that is,  $\sum_{ldru}
A^i_{ldru} \<ldru|= (\sum_{ldru} A^i_{ldru}\<uldr|) Y \otimes
Z\otimes Y^{-1}\otimes Z^{-1}$ for all $i$.
\end{thm}

\begin{figure}[h!]
 \centering
  \includegraphics[width=0.5\textwidth]{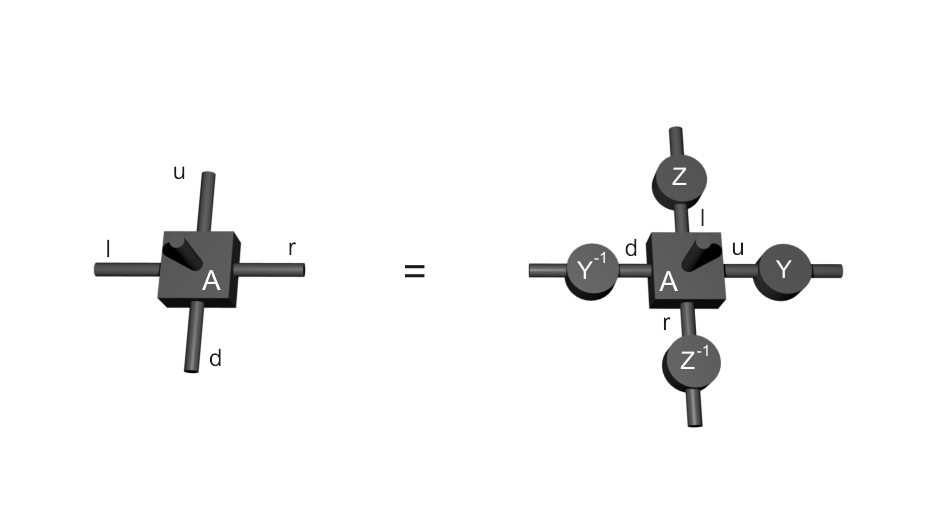}
    \caption{\footnotesize \emph{This figure represents the
  condition which must be fulfilled by a TI-PEPS in order to generate
  a state invariant under $\frac{\pi}{2}$-rotations (in this case a
  clockwise rotation).}}
 \label{fig12}
\end{figure}

In this case, one can see that $Y$, $Z$ must satisfy the additional constraints $(YZ)^T=YZ$, $(ZY)^T=ZY$.

Finally, we characterize the PEPS which are symmetric respect to a
$\pi$-rotation.

\begin{thm}[spatial $\pi$-rotation symmetry]\label{symm: pirot}
Let us consider an $L\times N$ TI-PEPS with the property that for a
region of size smaller than $L/5\times N/5$ it is injective and that
it is invariant under a  $\pi$-rotation, then there exist
invertible matrices $Y$, $Z$ such that the tensors defining the PEPS verify Fig. \ref{fig13}, that is,
$\sum_{ldru} A^i_{ldru} \<ldru|= (\sum_{ldru} A^i_{ldru}\<ruld|) Y
\otimes Z\otimes Y^{-1}\otimes Z^{-1}$ for all $i$.
\end{thm}

Now the constraints are $Z^T=Z$, $Y^T=Y$.

\begin{figure}[h!]
 \centering
  \includegraphics[width=0.5\textwidth]{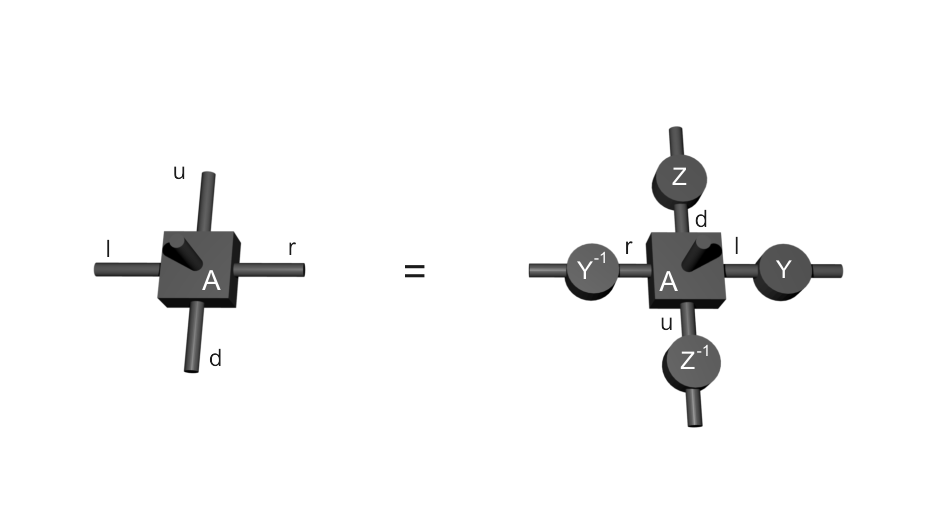}
  \caption[Fig. 8]{\footnotesize \emph{This figure represents the
  condition which must be fulfilled by a TI-PEPS in order to generate
  a state invariant under $\pi$-rotations (in this case a
  clockwise rotation).}}
 \label{fig13}
\end{figure}

\end{section}

\begin{section}{Applications}
It is clear that a symmetry must imposes restrictions on the possible behaviors and properties of a quantum system. Understanding these restrictions is a hard problem that has led the research in Quantum Many Body Physics in the last decades. For PEPS, which seem to provide a reasonably complete description of quantum states, we have proven a simple characterization of the existence of symmetries, which immediately leads to a number of consequences. In the lines below we list some of them.


\begin{subsection}{Lieb-Schultz-Mattis Theorem}

The Lieb-Schultz-Mattis theorem states that, for semi-integer spin, a $SU(2)$-invariant 1D {\it Hamiltonian} cannot have a uniform (independent of the size of the system) energy gap. This theorem has been generalized in a number of ways. Still in the 1D case but relaxing the symmetry to a $U(1)$ symmetry, Oshikawa, Yamanaka and Affleck showed in \cite{OYA97} that the same conclusion holds if $J-m$ is not an integer, where $J$ is the spin and $m$ the magnetization per particle. For the $SU(2)$ case in 2D, Hastings and Nachtergaele-Sims proved that the same results holds \cite{HastingsLSM}. In \cite{SWPC09}, we showed how the orginial Lieb-Schultz-Mattis Theorem can be understood on the level of states. More precisely, we showed that any $SU(2)$ invariant MPS with semi-integer spin cannot be injective. In this section we will give a 2D version of the Oshikawa-Yamanaka-Affleck theorem, by showing that a $U(1)$ symmetric PEPS for which $J-m$ is not an integer cannot be injective.

Let us start with a PEPS $|\psi_A\>$ of spin $J$ particles with
a $U(1)$ symmetry in the $z$ direction, that is
\begin{equation*}
u_g^{\otimes N}|\psi_A\>=e^{i\theta_{g}}|\psi_A\>
\end{equation*}
with $u_g=e^{igS_z}$. Since $g\mapsto e^{i\theta_{g}}$ is clearly
a representation, there exists $\theta$ such that
$\theta_{g}=Ng\theta$. We will show that
\begin{lem}\label{thetadef}
$\theta$ coincides with the magnetization
per particle $m$.
\end{lem}

To see this it is enough to expand both sides of the expression $u_g^{\otimes N}|\psi_A\>=e^{iNg\theta}|\psi_A\>$ around the identity:
from the LHS we get $u_g^{\otimes N} |\psi_A \rangle \simeq (\1 + i g
\sum_{j} S_{j}^{z}) |\psi_A \rangle$, while the RHS gives $(1+i N
\vec{g} \cdot \vec{\theta})|\psi_A \rangle$. By simplifying both
results, we get $\theta=\<\psi_A|\sum_{j} S_{j}^{z}|\psi_A \rangle$, the desired result.

Now we can prove the announced generalized Lieb-Schulz-Mattis
theorem for PEPS.

\begin{thm}\label{thm:GLSM}
Let us consider a PEPS $|\psi_A\>$ in a square $L\times N$ lattice that is injective in regions of size $L/5\times N/5$. If $|\psi_A\>$ is invariant under a representation of $U(1)$ with the usual generator of spin $J$ given by $S^{(J)}_z$, then the magnetization per particle $m$ fulfils that $(J - m)$ is an integer.
\end{thm}
If the state has full $SU(2)$ symmetry, then $m=0$ and we get the ``Lieb-Schultz-Mattis theorem'' for PEPS.

\begin{proof}
We will choose $R\ge L/5$, $S\ge N/5$ and consider the PEPS (with periodic boundary conditions) associated to the region $R\times S$, $|\psi_A^{R\times S}\>$. By injectivity it is clear that $|\psi_A^{R\times S}\>\not =0$. Applying $e^{igS_z^{(J)}}$ to all spins and using Theorem \ref{symm} we get that there must exist a choice of indices $k_{1},\ldots k_{RS}\in \{-J,-J+1,\ldots, J-1,J\}$ such that ${k_1}+\cdots +{k_{RS}}=SR\theta$. We do the same for regions of size $R\times (S+1),(R+1)\times S, (R+1)\times (S+1)$, getting indices $k',k''$ and $k'''$ respectively. Now
$$\theta=(R+1)(S+1)\theta -(R+1)S\theta- R(S+1)\theta +RS\theta=$$ $$=\sum_{r=1}^{RS} {k_r}+\sum_{r=1}^{(R+1)S}{k'_r}+\sum_{r=1}^{R(S+1)}{k''_r}+\sum_{r=1}^{(R+1)(S+1)}{k'''_r}\;.$$

The RHS has the same character as $J$, that is, it is integer if $J$ is and semi-integer if $J$ is. Therefore $\theta-J\in\mathbb{Z}$. Since, by Lemma \ref{thetadef}, $\theta$ is the magnetization per particle, we are done.
\end{proof}
\end{subsection}


\begin{subsection}{Wilson loops}
It has been observed in \cite{Verstraete} that the equal superposition of the four logical states of the toric code $|\psi\>$ has a PEPS representation with bond dimension $2$. Since the logical $X$ in the first (resp. second) logical qubit is implemented by a non-contractible cut of $\sigma_X$ operators along the vertical (resp. horizontal) direction \cite{Kitaev}, $|\psi\>$ remains invariant under these two ``Wilson loops'' (see fig. \ref{toric}).
\begin{figure}
  \includegraphics[width=8cm]{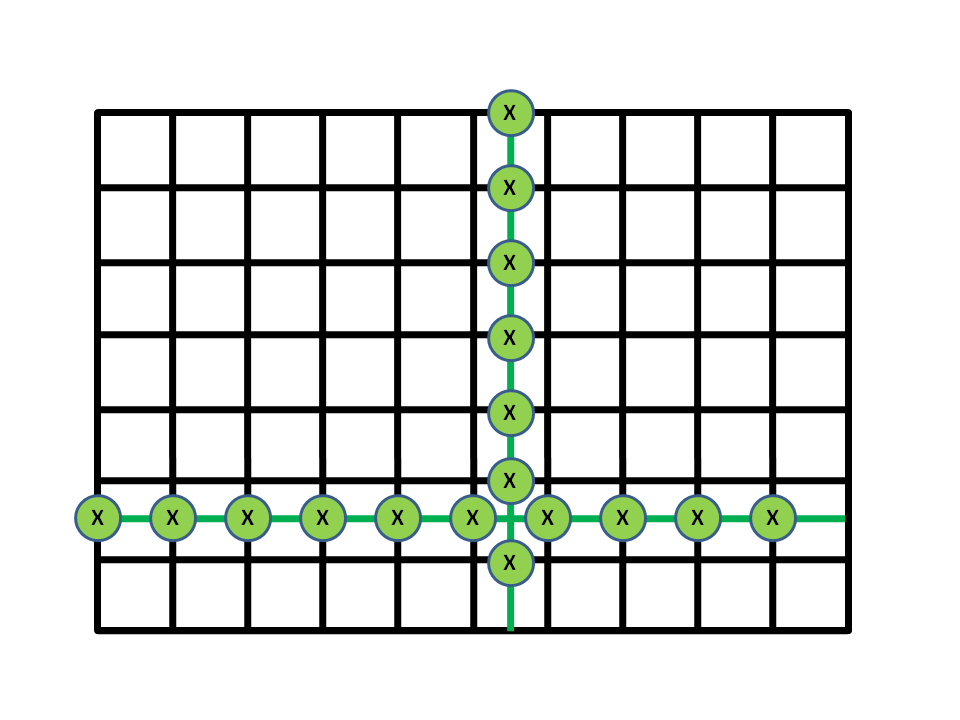}\\
  \caption{``Wilson loops'' that keep invariant the PEPS associated to the toric code.}\label{toric}
\end{figure}

We will see in this section how the existence of this kind of Wilson loops imply again that the PEPS cannot be injective.

\begin{thm}
Let $|\psi_A\>$ be PEPS in a $L\times N$ square lattice with local Hilbert space dimension $d$ such that there exists a $u\in U(d)$ with the properties:
\begin{enumerate}
\item[(i)] $u^{\otimes L}\otimes \1_{\text{rest}}|\psi_A\>=|\psi_A\>$ for a loop in the vertical direction.
\item[(ii)] $u^{\otimes N}\otimes \1_{\text{rest}}|\psi_A\>=|\psi_A\>$ for a loop in the horizontal direction.
\item[(iii)] $u\otimes \1_{\text{rest}}|\psi_A\>\not =|\psi_A\>$ for $u$ acting on a single site.
\end{enumerate}
Then $|\psi_A\>$ cannot be injective for any region of size $\le L/5\times N/5$.
\end{thm}
\begin{proof}
We assume injectivity for a region of size $L/5\times N/5$, (i) and (ii) and will show that (iii) does not hold. By applying (i) to all columns and Theorem \ref{symm}, we get that there exist unique $Y$ and $Z$ such that Fig. \ref{fig6} holds. Applying now (i) to $N/5$ columns and injectivity we get that $Y=\1$ and applying (ii) to $L/5$ rows and injectivity we get that $Z=\1$. So $u\otimes \1_{\text{rest}}|\psi_A\> =|\psi_A\>$ for $u$ acting on a single site.
\end{proof}

\end{subsection}
\end{section}


\begin{section}{Conclusions}
In this work we have provided a simple characterization of the existence of symmetries in PEPS. The result is based on the proven existence of a ``canonical form''. Since PEPS seem to  give a fairly complete characterization of the low energy sector of local Hamiltonians, the result paves the way for a better understanding of the restrictions that symmetries impose on quantum systems. As a first example of the kind of results that one can obtain from this characterization, we have shown a 2D version of the Oshikawa-Yamanaka-Affleck extension for $U(1)$ of the Lieb-Schultz-Mattis theorem. We have also outlined, via the injectivity property, how three of the main indicators of topological order (degeneracy of the ground state, existence of Wilson loops and corrections to the area law) are related.
\end{section}


\begin{section}{Acknowledgments}
M. Sanz would like to thank the QCCC Program of the
EliteNetzWerk Bayern as well as the DFG (FOR 635, MAP and NIM) for
the support. D. Perez-Garcia and C. Gonzalez-Guillen acknowledge financial support from
Spanish grants I-MATH, MTM2008-01366 and CCG08-UCM/ESP-4394 and M.M.
Wolf acknowledges support by QUANTOP and the Danish Natural Science
Research Council(FNU).
\end{section}

\end{document}